\documentclass[10pt,conference,letterpaper]{IEEEtran}

\usepackage{tikz}
\usepackage{amsmath}
\usepackage{filecontents}
\usepackage[english]{babel}
\usepackage{blindtext}
\usepackage{xcolor}
\usepackage{graphicx}
\usepackage{wrapfig}
\usepackage[justification=centering,skip=0pt]{caption}
\usepackage{subcaption}
\usepackage[normalem]{ulem}
\usepackage{mathtools}
\usepackage[linesnumbered,ruled,vlined]{algorithm2e}
\usepackage{courier}
\usepackage{booktabs}
\usepackage{multirow}
\usepackage{enumitem}
\usepackage{amsthm}
\usepackage[font=footnotesize]{caption}
\usepackage[backend=biber, style=ieee,minnames=1,maxnames=3]{biblatex}
\newtheorem{definition}{Definition}

\newtheorem{claim}{Claim}
\newcommand{\BfPara}[1]{{\noindent {\bf #1.}}}

\usepackage{titlesec}
\usepackage{eso-pic}
\usepackage{tcolorbox}
\usepackage{xcolor}

\newcommand\PlacePreprintNotice{
    \AddToShipoutPictureBG*{%
        \put(50,760){  
            \begin{minipage}{\textwidth}
            \centering
            \begin{tcolorbox}[colback=gray!10, boxrule=0pt, sharp corners, width=\textwidth, left=0pt, right=0pt, top=2pt, bottom=2pt]
                \small A version of this paper has been accepted for publication at IEEE INFOCOM 2025.\\
                © 2025 IEEE. Personal use of this material is permitted. Permission from IEEE must be obtained for all other uses.
            \end{tcolorbox}
            \end{minipage}
        }
    }
}
\AtBeginDocument{\PlacePreprintNotice}

\makeatletter
\newcommand{\linebreakand}{%
  \end{@IEEEauthorhalign}
  \hfill\mbox{}\par
  \mbox{}\hfill\begin{@IEEEauthorhalign}
}
\makeatother

\IEEEoverridecommandlockouts
\IEEEaftertitletext{\vspace{-2\baselineskip}}

\begin{document}

\title{Tree embedding based mapping system for low-latency mobile applications in multi-access networks \vspace{-.5em}}

\author{
    Yu Mi\IEEEauthorrefmark{1}, Randeep Bhatia\IEEEauthorrefmark{2},
    Fang Hao\IEEEauthorrefmark{2}, An Wang\IEEEauthorrefmark{1}, 
    Steve Benno\IEEEauthorrefmark{2}, Tv Lakshman\IEEEauthorrefmark{2} \\ 
    \IEEEauthorrefmark{1}Case Western Reserve University, OH, USA\\
    \IEEEauthorrefmark{2}Nokia Bell Labs, NJ, USA \\
    {\small Email: \IEEEauthorrefmark{1}\{yxm319, axw474\}@case.edu,  
    \IEEEauthorrefmark{2}\{randeep.bhatia, fang.hao, steven.benno, tv.lakshman\}@nokia-bell-labs.com} \\
}

\IEEEoverridecommandlockouts
\IEEEpubid{\makebox[\columnwidth]{XXX-X-XXXX-XXXX-X/25/\$31.00~\copyright2025 IEEE \hfill} \hspace{\columnsep}\makebox[\columnwidth]{ }}

\maketitle

\begin{abstract}
    Low-latency applications like AR/VR and online gaming need fast, stable connections. New technologies such as V2X, LEO satellites, and 6G bring unique challenges in mobility management. Traditional solutions based on centralized or distributed anchors often fall short in supporting rapid mobility due to inefficient routing, low versatility, and insufficient multi-access support. In this paper, we design a new end-to-end system for tracking multi-connected mobile devices at scale and optimizing performance for latency-sensitive, highly dynamic applications. Our system, based on the locator/ID separation principle, extends to multi-access networks without requiring specialized routers or caching. Using a novel tree embedding-based overlay, we enable fast session setup while allowing endpoints to directly handle mobility between them. Evaluation with real network data shows our solution cuts connection latency to 7.42\% inflation over the shortest path, compared to LISP's 359\% due to cache misses. It also significantly reduces location update overhead and disruption time during mobility.
\end{abstract}

\section{Introduction}

In highly dynamic environments characterized by frequent location changes, establishing and maintaining stable, low-latency connections can be a significant challenge. Traditional mobility protocols, such as MIPv6 (Mobile IPv6)~\cite{rfc6275} and PMIPv6 (Proxy Mobile IPv6)~\cite{rfc5213}, typically direct all communication paths of a mobile node (MN) through a single anchor point, creating lengthy detours that result in substantial delays. This issue is particularly problematic in applications like V2X (Vehicle-to-Everything) communication, where hosts can rapidly move significant distances, causing them to quickly become distant from their respective anchors~\cite{V2X}, and, as a result, adding significant latency.
Moreover, in applications like LEO (Low Earth Orbit) satellite communication, anchors placed on satellites can move quickly away from the host~\cite{LEO-SURVEY}, leading to increased latency and reduced connection stability. Ground-based anchors can also result in inefficient routing due to the need for an extra round trip between the satellite and ground stations.

One potential solution is to use distributed mobility management (DMM), which allows the MN to be allocated a new close by anchor at the point of attachment~\cite{rfc8818,rfc8885}. To maintain session continuity while moving across different networks, the old anchors either tunnel the existing packet flows to the new anchor, incurring high latency; or they continue to forward the existing packet flows to the MN, incurring high latency and excessive overhead for maintaining multiple anchors for the MN at the same time~\cite{IEEE-DISTRIBUTED-MOBILITY}.

Partial or full DMM approaches, such as LISP (Locator/ID Separation Protocol)~\cite{rfc9300,rfc9301},  offer a solution to the session continuity issue through the separation of end-user device identifiers and the routing locators used to reach them. 
LISP utilizes a database that maps  identifiers to  routing locators, with caching of this mapping at multiple ingress routers or network nodes. 
This method enables direct forwarding to the device on low-latency paths.
However, LISP has limitations during cache misses and updates.
Cache misses in LISP can result in increased session setup delay, while cache updates can pose scalability concerns, since device movement in proximity can trigger update to a large number of remote caches, creating large disruptions.

In today's multi-connected mobile devices, there are numerous opportunities to optimize routing through dynamic steering and switching of application traffic across diverse access networks (cellular, Wi-Fi, wired).
However, existing mobility management solutions, which were initially designed for single access technologies, cannot achieve this effectively. 

To address these challenges, we propose a novel end-to-end approach for managing mobility of multi-connected mobile devices at scale to optimize performance for latency-sensitive applications.
Our approach extends the Locator/ID separation idea from LISP to support multi-access and introduces a distributed scheme for mapping identifiers to locators for individual accesses.
Our design has two key differences from LISP: 
i) an optimized Locator LookuP (LLP) service using a new tree-embedding algorithm is provided through a low-latency network overlay, enabling fast connection setup with on-the-fly location discovery;
ii) the mobility function is supported at end devices as an added software module, which eliminates the need to modify existing routers and avoids caching and the associated cache maintenance overhead.

Compared to existing solutions, our design has the following unique advantages:
i) It eliminates inefficient triangular routing while maintaining session continuity. 
Our tree-embedding algorithm ensures the initial setup messages are forwarded to the destination by the LLP overlay on an optimized path that is only slightly longer than the direct least  latency path.
Moreover, all subsequent messages are forwarded  on the direct  least latency path between the devices.
ii) Location updates are handled efficiently and swiftly.
With our tree-embedding algorithm, only a few nodes within the LLP overlay typically need to be updated when the location or access connectivity of a device changes.
iii) Allowing end hosts to perform access selection and switching  enables a wide range of complex flow management policies for selecting and maintaining the best low-latency route. The selection can be individual for each connection in each direction, regardless of the underlying network technology, allowing for greater flexibility and improved performance. 
iv) Not requiring router change make the solution easier to be deployed incrementally.
Mobility modules can be downloaded to devices as part of a software update, and LLP overlay  nodes can be gradually deployed across different Point-of-Presences (PoP).

We have evaluated our solution by using real network topology and latency data. 
Results show that our solution significantly reduces connection latency, location update overhead, and connection disruption time over existing solutions.
For instance, on average, our connection setup latency is only 7.42\% higher than the direct path latency, compared to 359\% for LISP.
Likewise, our solution cuts down the resource cost for storing and maintaining location updates by 4 times and reduces connection disruption time under mobility by 2.6 times, compared to LISP.

This paper unfolds as follows: 
Section~\ref{sec:solution_description} presents our overall solution.
Section~\ref{sec:design:llplocation} describes the LLP tree construction algorithm, the key component of our design.
Section~\ref{sec:implementation} shows the implementation ideas.
Section~\ref{sec:eval} describes the performance evaluation.
And finally we discuss related work and conclusion.

\section{Solution Description}
\label{sec:solution_description}

\begin{figure}[h!]
\vspace{-2em}
\centering
    \begin{subfigure}{0.40\columnwidth}
        \includegraphics[width=\textwidth]{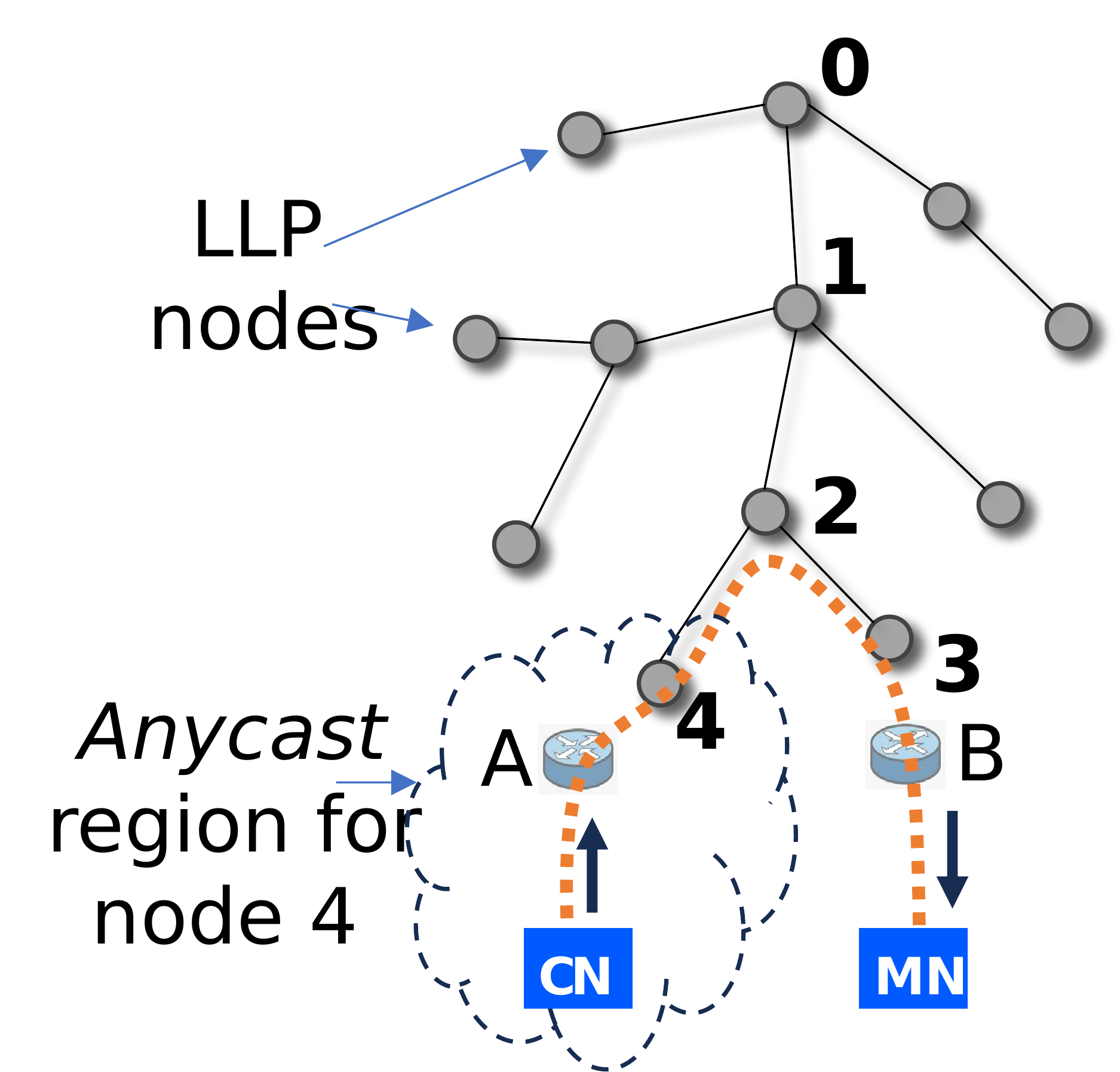}
        \caption{CSR Forwarding on LLP tree}
        \label{fig:LLP-tree-basic:A}
    \end{subfigure}
    \rule{0.1pt}{4.cm} 
    \begin{subfigure}{0.30\columnwidth}
        \includegraphics[width=\textwidth]{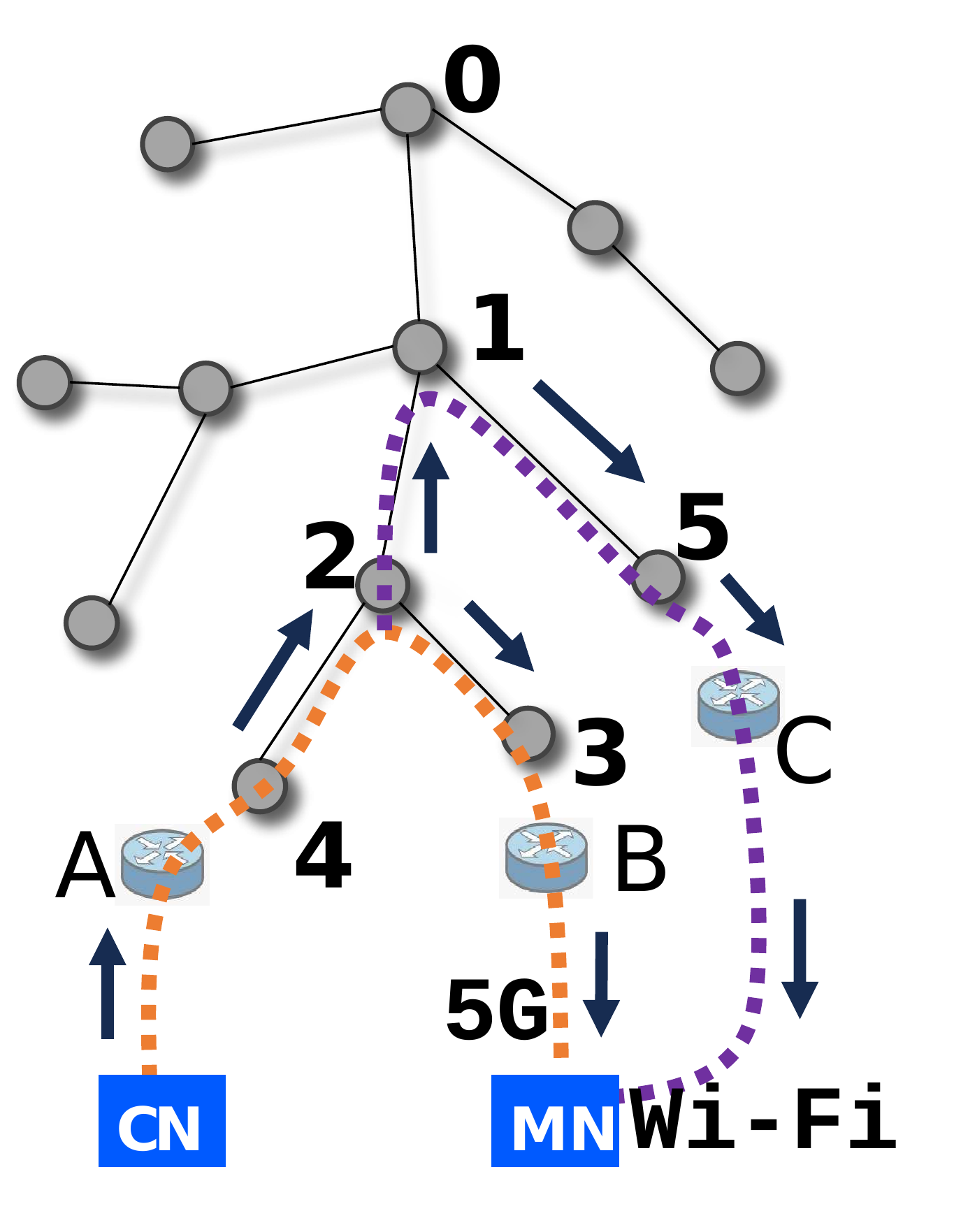}
        \caption{CSR Forwarding to Multi-access}
        \label{fig:LLP-tree-basic:B}
    \end{subfigure}
    \rule{0.1pt}{4.cm} 
    \begin{subfigure}{0.23\columnwidth}
        \includegraphics[width=\textwidth]{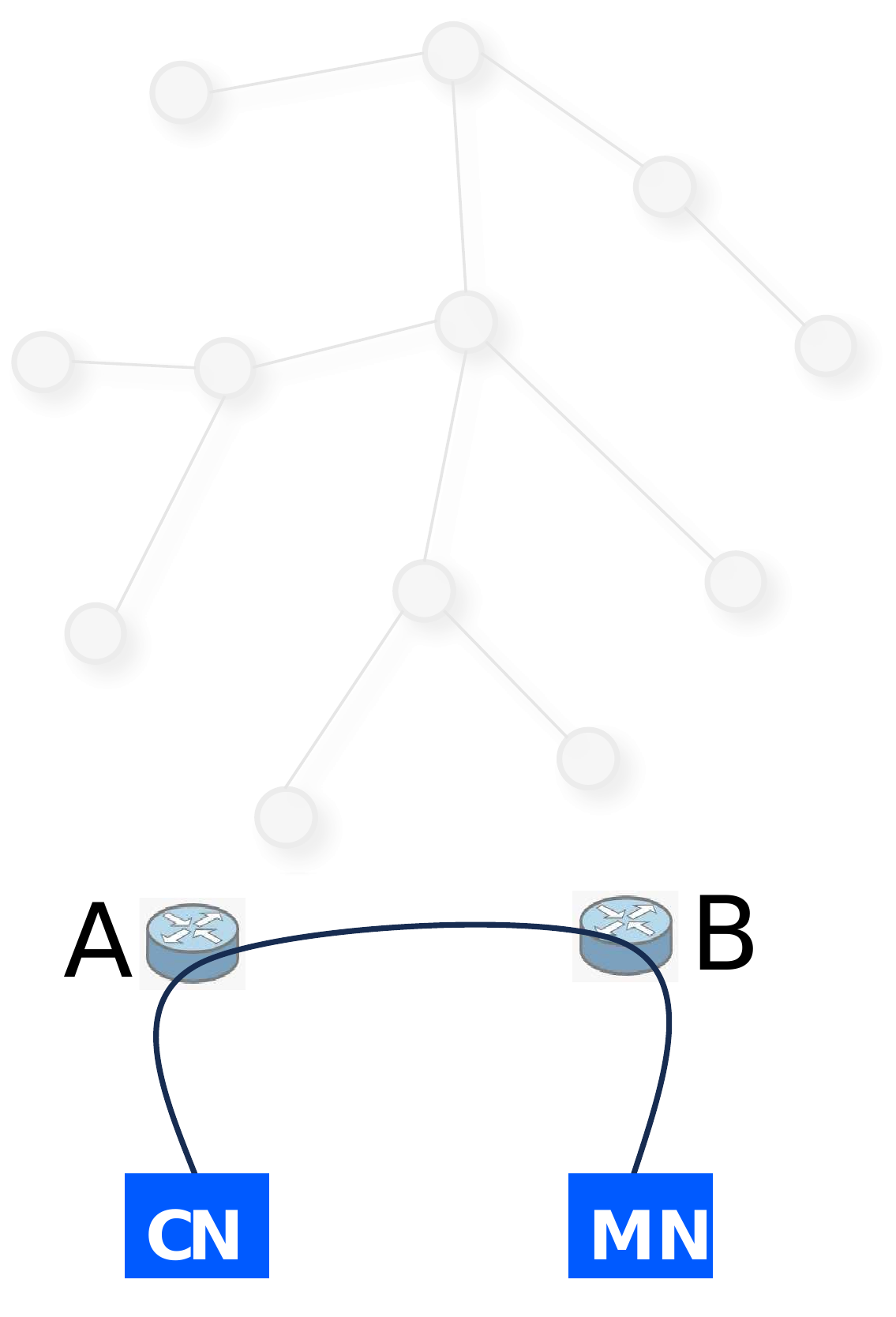}
        \caption{Peer-to-Peer Connection}
        \label{fig:LLP-tree-basic:C}
    \end{subfigure}
    \vspace{.5em}
    \caption{CSR Forwarding}
  \label{fig:LLP-tree-basic}%
\vspace{-1em}
\end{figure}

\label{sec:procedures}
Using physical IP addresses (PIPs) that are assigned by the access links as identifiers for mobile devices would bring difficulties in maintaining accessibility and session continuity, especially for multi-access devices with multiple independent PIPs for each access network (5G, Wi-Fi, Ethernet, etc.).
Inspired by the concept of Locator/ID separation, we assign each Mobile Node (MN) a globally unique IP address (GIP), which serves as its fixed identifier, despite its location.
We maintain the mapping of an MN's GIP to each of its PIPs for its different accesses.
This guarantees that connection setup requests (CSRs) addressed to its GIP can be delivered through any of its accesses when requested.

The mapping from a GIP to its PIPs is maintained by a locator lookup (LLP) overlay deployed at a sufficiently large number of selected Point-of-Presences (PoPs) across the IP backbone network.
This allows each Corresponding Node (CN) or Mobile Node (MN) to have access to a nearby LLP node.
The LLP overlay maintains the GIP to PIP mapping for each MN, and supports delivery of initial connection setup requests (CSR) from CN to MN (Figure~\ref{fig:LLP-tree-basic:A}).
For multi-access MNs, the LLP overlay ensures that the CSR is forwarded to the MN on each of its registered access networks (Figure~\ref{fig:LLP-tree-basic:B}).
Note that typically, only the CSR is forwarded through the overlay; all subsequent communication between the CN and MN use the direct path between them (Figure~\ref{fig:LLP-tree-basic:C}).

Our goal is to design an efficient LLP overlay to support location update and CSR forwarding procedures, assuring:
i) only a small number of LLP nodes maintain the mapping for each GIP,
ii)  only a small number of LLP nodes (near the MN) require updating when an MN's GIP to PIP mapping changes due to localized mobility,
iii) locator lookup and CSR forwarding can take place simultaneously, and
iv) CSRs follow optimized paths of low latency, that are only slightly longer than the direct least  latency path.

The design of LLP overlay is the main focus of this work.
The overlay is organized in an optimized {\em tree} structure.
Its leaf nodes  store GIP to PIP mapping entries for nearby mobiles, while its internal nodes maintain entries that facilitate  forwarding of messages to leaf nodes, where the GIP to PIP mappings are kept.
These entries are consistently updated to reflect the current  location of the mobile device, as explained  below.

\subsection{Location Updates on the LLP overlay}

\begin{figure}[h!]
\vspace{-1.5em}
\centering
    \begin{subfigure}{0.35\columnwidth}
        \includegraphics[width=\textwidth]{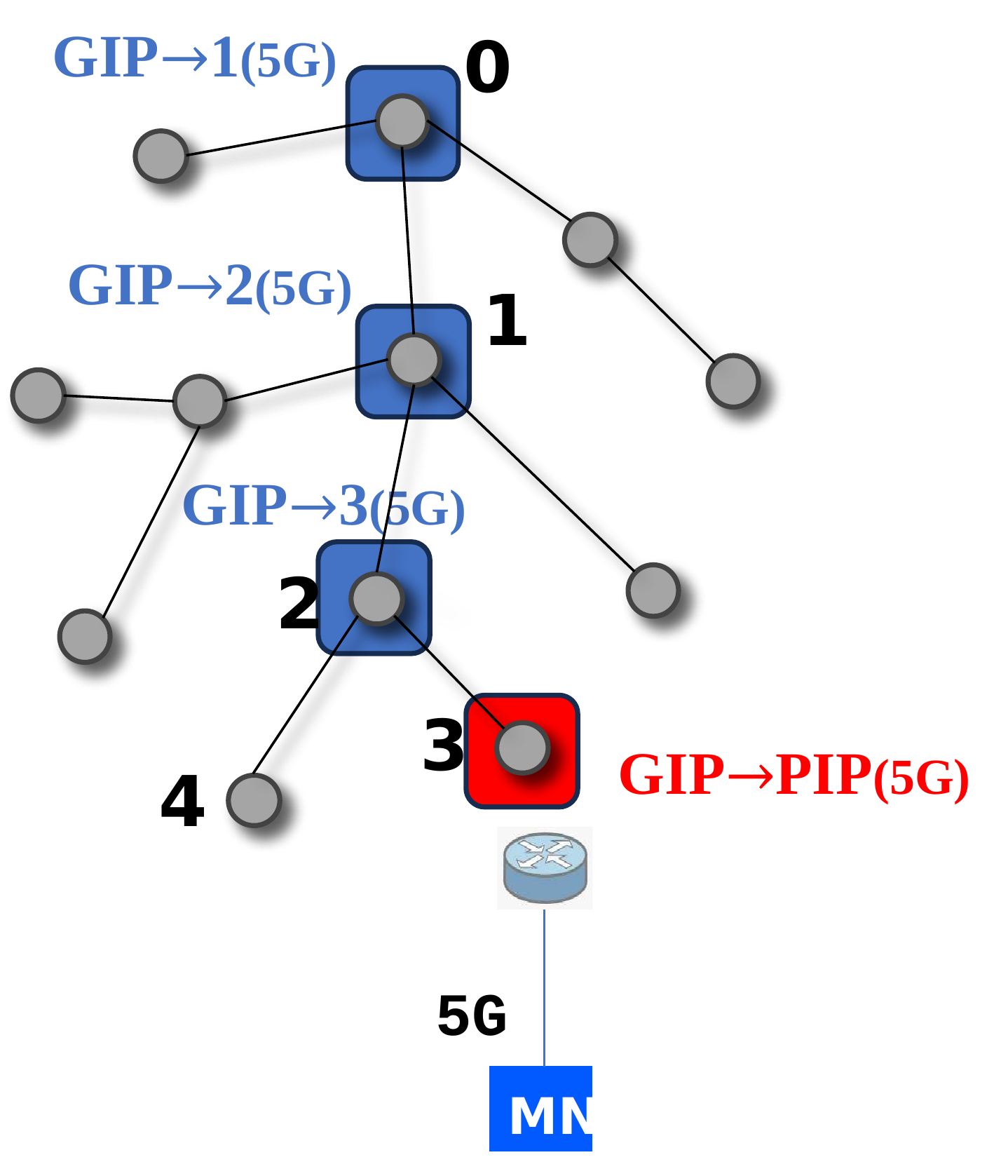}
        \caption{MN at initial location}
        \label{fig:location-tracking:A}
    \end{subfigure}
    \rule{0.1pt}{4cm} 
    \begin{subfigure}{0.3\columnwidth}
        \includegraphics[width=\textwidth]{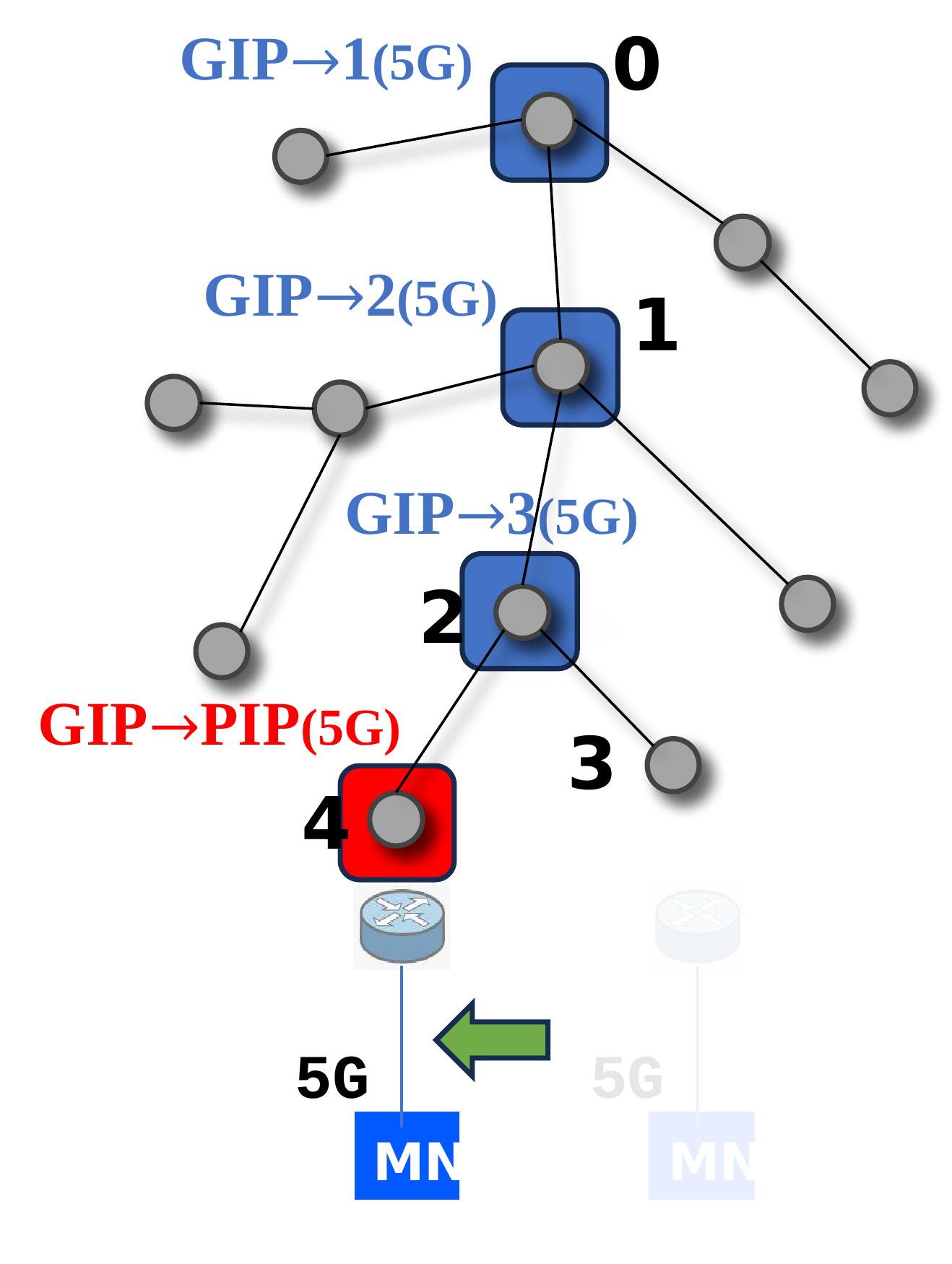}
        \vspace{-2em}
        \caption{MN moves to a new location}
        \label{fig:location-tracking:B}
    \end{subfigure}
    \caption{Location Update Handling  in LLP Tree}
    \vspace{-1em}
  \label{fig:location-tracking}%
\end{figure}

Location update can be triggered due to the MN's registration in a new access network (Figure~\ref{fig:location-tracking:A}) or change in its location or PIP from an already registered access network (Figure~\ref{fig:location-tracking:B}).
A location update contains the MN's GIP and its latest PIP for the access. 
When the nearby LLP leaf node receives the update, it stores the GIP to PIP mapping for the specified access, and propagates the location update up the tree to ensure the GIP entries at the parent nodes are updated as needed, and the outdated entries on the path down to the prior leaf node, if any, are deleted.

Figure~\ref{fig:location-tracking:A} illustrates the initial location registration.
As the location update goes up the tree, the leaf LLP node $3$ stores the GIP to PIP mapping for 5G (red box).
While other nodes $2$, $1$, $0$ only stores its child who knows how to reach the 5G PIP (blue boxes).
For example, the blue box at node $2$, labeled  $GIP\rightarrow 3(5G)$, indicates that the next hop for the GIP for 5G is node $3$.
It is obvious that the entries for this GIP are only maintained at nodes along the tree path from the leaf node $3$ to the root node  $0$. 

Figure~\ref{fig:location-tracking:B} illustrates how GIP to PIP mappings are updated when the MN changes location.
In this example, the MN's 5G access moves from leaf node $3$ to node $4$.
The node $4$ receives the location update, adds a new entry  $GIP\rightarrow PIP(5G)$, and forwards the update to its parent node $2$.
Node $2$ then modifies the corresponding GIP entry from $3$ to $4$, and informs node $3$ to delete the stale GIP entry.
Note that no other nodes are affected by this location update besides nodes $4, 2, 3$. 

To enforce the updates' reliability, the last LLP node at each update path, i.e., the node that has no other node to forward to (e.g., node $0$ in Figure~\ref{fig:location-tracking:A} or node $4$ in Figure~\ref{fig:location-tracking:B}) is responsible for sending an ACK message back to the MN; and re-transmission can be triggered as needed.

The LLP node GIP entry also includes a {\em total access count}, the total number of unique leaf nodes that contain PIP mappings for this GIP.
Its maintenance is done through a slightly different update procedure, omitted due to space constraints.

\subsection{Mapping and Connection Forwarding}

As illustrated in Figure~\ref{fig:LLP-tree-basic:A}, a Corresponding Node (CN), seeking to establish a connection with a Mobile Node (MN) for which it has no GIP-to-PIP mappings, sends the Connection Setup Request (CSR) packet addressed to the MN's GIP. This CSR then reaches a nearby Leaf LLP node $4$ via anycast. In this example, we assume nodes $3,2,1,0$ store the GIP mapping for this MN as shown in Figure~\ref{fig:location-tracking:A}.

The leaf LLP $4$ first checks its local entries for a GIP-to-PIP mapping.
If the Leaf LLP node has at least one entry for the GIP, it can resolve the PIP for each entry, and forwards the CSR to the MN using its PIP for the corresponding access.
If there is no match, or the Leaf LLP node does not have entries for all accesses (determined based on the total access count of this GIP), it sends the CSR  to its parent node for further processing.
In this example, there is no match at node $4$, so the CSR is forwarded to node $2$. 

The parent node $2$ processes the CSR similarly: for each GIP entry it has, it forwards the CSR to its children listed as next hop forwarding nodes in those entries.
If there is no match, or the parent node misses some access entries for the GIP within its subtree, the parent node continues to propagate the CSR up the tree.
In this example, the node $2$ finds one matching pointing to its child node $3$ and the total access count is $1$, hence it only needs to forward CSR to its child $3$. 

When receiving a CSR from the parent node, the child node in turn forwards the CSR to its child, which is the next hop for this GIP.
The process continues until the CSR reaches the leaf node, where the PIP for the GIP is resolved, and the CSR is forwarded to the MN using its PIP for the corresponding access.
In Figure~\ref{fig:LLP-tree-basic:A}, node $3$ resolves the GIP-to-PIP mapping and delivers the CSR to the MN.

Figure~\ref{fig:LLP-tree-basic:B} shows an example where the MN has two different accesses: 5G and Wi-Fi, with their gateways near LLP nodes $3$ and $5$ respectively.
Since the node $2$ lacks GIP entries for Wi-Fi, which are maintained on the path $5,1,0$, it makes a copy of the  CSR and sends it to its parent node $1$.
From there, the CSR eventually reaches node $5$ (purple path), where the GIP-to-PIP resolution for Wi-Fi occurs, and the CSR is forwarded to the MN over Wi-Fi.
As a result, the mobile receives a copy of the CSR on both its 5G and Wi-Fi interfaces, but only needs to accept the first copy.

Once the connection is established, communication occurs directly between the MN and CN, with mobility updates communicated directly from the MN to the CN, bypassing the LLP system (Figure~\ref{fig:LLP-tree-basic:C}).

\subsection{Further discussion}
We summarize the advantages of the LLP based solution below, and also briefly discuss several other design considerations that are not covered in detail due to space limitation.

{\noindent \bf Advantages of LLP approach:}
i) Each MN's location for access is tracked by only a few LLP nodes (equals the depth of the LLP tree, which grows at a logarithmic rate, as we show later).  
ii) As shown in Section~\ref{sec:eval}, the average latency of tree paths,  for forwarding CSRs,  is close to  best possible.
iii) Location updates are processed quickly, as they primarily result in modifications to entries  within a limited subset of close by LLP nodes, particularly in cases of local mobility.

{\noindent \bf Scalability:}
Scaling to millions of end devices can be achieved by partitioning GIPs by prefixes and creating separate overlays per GIP partition, allowing the LLP tree to scale horizontally.

{\noindent \bf Security:}
To ensure the integrity of location updates and avoid connectivity to each MN being hijacked, signatures can be inserted in location update packet headers, e.g., using the HMAC field in SRv6 header. 
No additional overhead is incurred in the tree except for signature processing.
Similar mechanism can be also be used for CSR for added security.

{\noindent \bf Reliability and dynamic conditions:}
LLP nodes can be engineered using high availability clusters to ensure server reliability. Since the tree is an overlay, impact of underlay network dynamics such as network failures and route changes can be minimized by
traffic engineering of the underlay.
\section{LLP tree building algorithm}
\label{sec:design:llplocation}
We now present a Hierarchical Cluster Selection (HCS) {\em LLP overlay  tree} construction algorithm. 
HCS takes a given network of Point of Presence (PoP) sites as input and generates an overlay tree $T$ consisting of PoP sites selected by HCS for LLP node deployment and the least latency direct paths on the underlay network of PoP sites as links.
HCS offers two significant advantages: (i) $T$ has a low average latency increase compared to the underlay network, and (ii) it can effectively limit the number of required LLP nodes.

The HCS algorithm is inspired by a metric tree embedding method~\cite{metric2003fakcha, bartal1998}, but augmented with novel cluster center selection and shortcut-based tree augmentation heuristics. These improvements result in significant performance enhancements in practice, as we show in our evaluations.

Let  $G = (V,E)$ be the fully connected representation of the underlay network of PoP sites, where $V$ is the set of $n$ PoP sites and each undirected link $e = (u,v) \in E$ corresponds to the least latency path between sites $u$ and $v$ in the underlay. $l(u,v)$ denotes the latency of this path.
HCS forms a hierarchical clustering $C$ of the nodes in $V$. The clusters in $C$ have a parent-child relationship, which naturally lends itself to a tree hierarchy. The LLP tree $T$ is constructed from these clusters by selecting a center node for each cluster $c \in C$ and attaching a child cluster to its parent cluster  using a link of $E$ between their center nodes. 

\vspace{-.5em}
\begin{algorithm}
\caption{HCS—Hierarchical Cluster Selection }
   \SetKwProg{Fn}{function}{}{}
    \SetKwFunction{Myalgo}{HCS}

    \Fn{\Myalgo{$S,D,G,lt, \pi$}}{ 
    \If{\(\mbox{Bounded}(S,D,G,lt)\)}{
            \KwRet CreateLeaf(S);
        }
    $T \gets$ \mbox{Root}($S$)\;
    $C \gets$ \mbox{RandomClustering}($S,D/\alpha,G,lt,\pi$)\;
    \ForEach{$c$ \textbf{in} $C$}{
    AttachChild($T$, \Myalgo{\(c,D/\alpha,G,lt,\pi\)})\;
}
\KwRet($T$)
 }
\end{algorithm}
 \vspace{-1em}
HCS utilizes  parameters $S, D, lt, \alpha$, and $\pi$:

\begin{itemize}
\item $S$: The set of cluster nodes, initially set to $V$.
\item $D$: A variable representing the maximum latency used for clustering. It is initialized to the latency diameter of $G$, denoted as $D(G) := \max(l(u, v))$, where $u,v \in V$.
\item $lt$: An acceptable latency threshold. Once the latency diameter of a cluster falls below this threshold, no further clustering is required.
\item $\alpha > 1$: A parameter that controls both the tree depth and its branching, each of which may impact latency. In practice, we have found that a good choice is $\alpha = 2$.
\item $\pi$: A random ordering of the $n= |V|$ nodes.
\end{itemize}

HCS  just returns a single leaf cluster {\em {CreateLeaf($S$)}},  if the diameter of the latency graph of the nodes in $S$ is at most the latency threshold $lt$.  This diameter check is performed by  {\em Bounded} function in line $2$. Otherwise, it  randomly partitions nodes of $S$, into a set of clusters $C$. This partitioning is performed by the {\em RandomClustering} function in line $5$ that operates with a cluster latency threshold $CLT$ (its second argument, which is set to $D /\alpha$ in line $5$). 

 {\em RandomClustering}   selects the first node $v_1$ from the randomly ordered list $\pi$ of nodes, and gathers all the nodes in $S$ that are reachable within the given latency threshold $CLT$ latency from $v_1$. It marks  the gathered nodes as “resolved'' and groups them into a single cluster. It then moves on to the next node $v_2$ of $\pi$ and performs the same gathering and grouping operation on the remaining unresolved nodes in $S$ to form the second cluster. This process is repeated until all nodes in $S$ have been assigned to a cluster. To form the final clusters, it assigns a center to each cluster by selecting a “central” node that has the least total latency to all other nodes in the cluster.  Finally, it returns the set of all these clusters (denoted  $C$ in line $5$).

In lines $6-7$, the HCS function is recursively called to perform hierarchical clustering on the nodes within each cluster $c$ in $C$, but this time with a latency threshold of  $CLT/\alpha$. The resulting set of hierarchical clusters becomes the children of the root cluster formed by the nodes in $S$. This is accomplished by the {\em AttachChild} function, which adds the set of edges $(r, r_i)$ of $E$ that connect the center $r$ of the root cluster with the center $r_i$ of the $i$-th child cluster. This process results in a tree of hierarchical clusters. The centers of these clusters and the added edges constitute the nodes and edges of the LLP tree $T$, which is returned in line $8$.

\begin{figure}[h]
\vspace{-1em}
  \centering
  \begin{subfigure}{0.37\columnwidth}
    \centering
    \includegraphics[width=.8\linewidth]{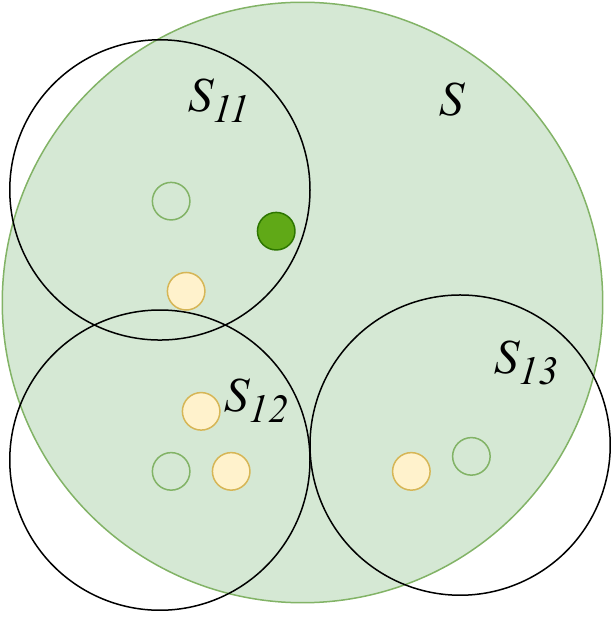}
    \caption{}
    \label{fig:tree:cover}
  \end{subfigure}%
  \begin{subfigure}{0.57\columnwidth}
    \centering
    \includegraphics[width=.8\linewidth]{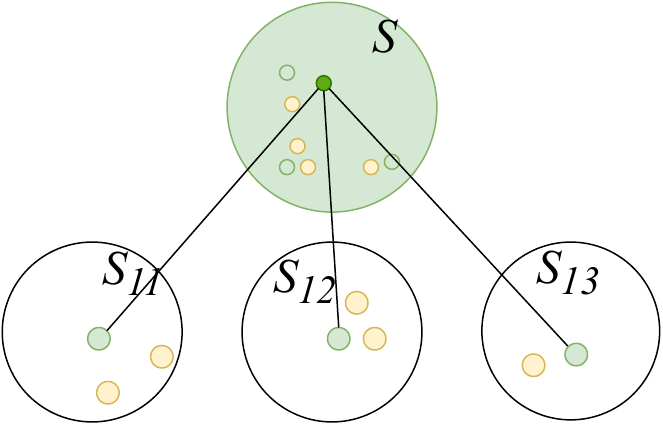}
    \caption{}
    \label{fig:tree:dia}
  \end{subfigure}
  \caption{Tree Building Example}
  \label{fig:tree}
\vspace{-1.em}
\end{figure}

Figure \ref{fig:tree} provides an example of the tree-building algorithm. 
Figure \ref{fig:tree:cover} shows the distribution of the nodes in space and how they are hierarchically clustered for $\alpha=2$. Initially, the set $S$, which is a cluster of all nodes, has a very large diameter and is therefore further partitioned using  {\em RandomClustering}  which yields the $3$ clusters $S_{11}, S_{12}, S_{13}$. When HCS is recursively called on these clusters, {\em Bounded}  returns a leaf cluster for each of them, since their latency diameters are less than the latency threshold $lt$. This results in the tree of clusters shown in  Figure \ref{fig:tree:dia}, with one root cluster for $S$ and these $3$  clusters as leaves.
In these figures, the green  nodes are cluster centers.

\subsection{Analysis}

\begin{definition}
\label{inflation}
{\bf Latency inflation} $LI(u,v)$ for nodes $u$ and $v$:  A measure of the increase in latency from node $u$ to node $v$ for going through $T$ versus taking a direct connecting path.
\end{definition}

Let $l_T(u,v)$ and $l_G(u,v)$ be the latencies of the best paths from $u$ to $v$ in $T$ and $G$, respectively. Then:
\begin{equation}
    LI(u,v)=\left[ \frac{l_T(u,v)}{l_G(u,v)} -1 \right] 
\end{equation}

$LI(u,v)$ is non-negative and ideally should be  close to $0$.

In the following, we take $\alpha=2$.
\begin{claim}
\label{claim:height}
The depth $H$ of $T$ is $ O( \log_2 {\frac{D(G)}{lt}})$.
\end{claim}
\begin{proof}

Note that with every recursive call, the cluster latency threshold   $CLT$ goes down by a factor $\alpha=2$. So, at the $i$-th recursive call $CLT = D(G)/2^{i}$. Note also that at the $i$-th recursive call,  two nodes $u$ and $v$ are grouped together into one cluster $c$  only if their  latency to some common node $w$ (in $\pi$) is within this CLT value. By triangle equality (since  latencies in the fully connected graph $G$ form a distance metric), it follows  that latency between $u$ and $v$, and hence latency diameter of $c$, is at most $\min \{ D(G), 2CLT \}$.  
When $i = \log_2 {\frac{D(G)}{lt}}+ 1$, $2CLT \le lt$ and hence in the next recursion, $c$ is marked as a leaf cluster and the recursion stops. Since the recursion depth is a measure of the tree depth, the result follows.
\end{proof}

\begin{claim}
Latency of level $i$ tree edges  is at most $lt 2^ i $.
\label{latency-level}
\end{claim}
\begin{proof}
From proof of   claim~\ref{claim:height} it follows that at recursion depth  $H -k$, the latency diameter (LD) of the cluster $S$ is below $lt 2 ^k$. Since a level $k$ edge  is added between two nodes that both belong to the same cluster $S$ at recursion depth $H-k$, its latency is at most  $lt 2^k$ and the result follows.
\end{proof}

The following result follows from above result and  result of ~\cite{metric2003fakcha, bartal1998} (proof  omitted).
\begin{claim}
\label{claim:inflation}
The average latency inflation of $T$ is  $O(\log n)$.
\end{claim}
\noindent where $n$ is the number of PoP sites.

HCS  determines the placement of LLP nodes at the central nodes of clusters, as this method usually results in lower latency. However, there may be other cluster center choices that can further reduce latency. To explore these alternatives, we  consider heuristics in the next section that select cluster centers based on criteria beyond mere closeness to other nodes within the cluster.

\subsection{Heuristics for cluster center selection}
\label{improve-tree}

\subsubsection{Selecting cluster centers on the basis of both inter and intra cluster latencies}
It is crucial for cluster centers to not only be  close  to  other nodes within the cluster but also to their parent cluster. This is because all paths between nodes within a cluster and those outside of it utilize the edge connecting the cluster's center to its parent. We  use this insight to improve the choice of cluster centers. Specifically, we  select the node of the cluster to serve as its center by minimizing the following quantity:

 \begin{equation}
    \arg \min_{i} \frac{  \sum_{j \in c, j \ne i}   {w_{j} l_{G}(i,j) }   }    { \sum_{j \in c}   w_{j} } + l_G(i,p)
\end{equation}. 

Here, $c$ denotes a cluster, $p$ is the center of its parent cluster, and $w_j$ is the weight for node $j$. These node weights are assigned based on the traffic or population density at the PoP site representing that node. The first sum represents the weighted average of the intra-cluster latencies for the choice of node $i$ as the cluster center, while the second term is simply the latency to the parent LLP node for that choice. By accounting for both the inside and outside cluster latencies for the LLP tree, this method yields a better choice for cluster centers.

This method of picking centers can be applied level by level, starting from the root and moving down to the leaves of the tree of clusters (Figure~\ref{fig:tree:dia}). This way, the center of the parent cluster is already selected and does not change after the method is applied to a child cluster.

\subsubsection{Stretching Detours}

The paths in the overlay LLP tree when mapped to the underlay may contain detours.

To address this issue, we design a heuristic  that makes adjustments to tree centers to eliminate unnecessary detours and as a consequence 
further improves the latency of tree paths.
It repeatedly picks any two clusters $c_1$ and $c_2$, where $c_2$ is a parent of $c_1$ in the tree, such that  the lowest latency underlay path from $c_1$'s center   $n_1$   to $c_2$'s center   $n_2$ includes  $n_3$, another node in $c_1$, which is itself a center of another cluster $c_3$.
It switches the center of cluster $c_1$ from    $n_1$ to  $n_3$, thus eliminating this detour. After this switch, the HCS algorithm is invoked to re-build the subtree under cluster $c_1$ and the process continues.

\subsection{Shortcut based Tree Augmentation}
\label{sec:shortcut}

As we have shown,  the average latency inflation of the overlay LLP tree constructed by HCS is small  (see Claim~\ref{claim:height},~\ref{claim:inflation}). Now we show how to achieve low worst case latency inflation using the LLP tree by augmenting it with a few “shortcuts”, to bypass its high latency paths.

\begin{wrapfigure}{r}{0.35\columnwidth}
\vspace{-1.5em}
\centering
  \includegraphics[clip, width=\linewidth]{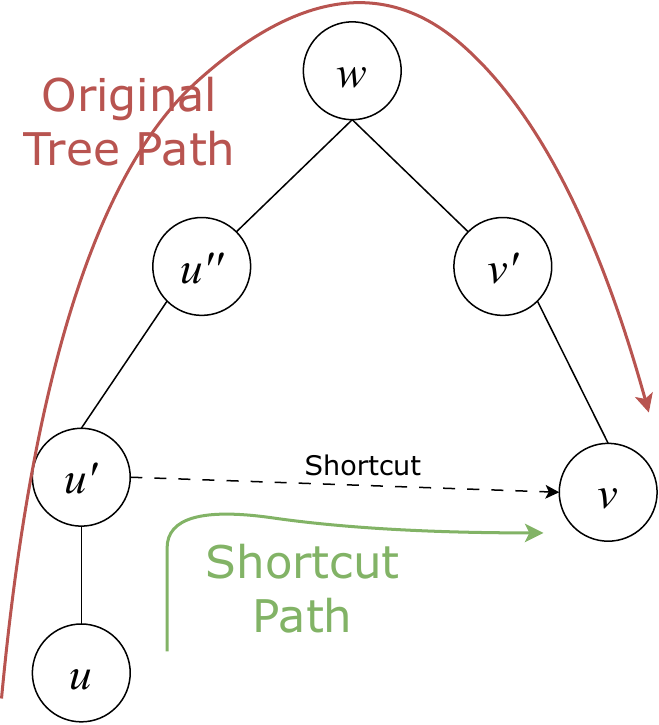}
  \caption{LLP Tree Shortcut Example}
  \label{fig:shortcut_example}
  \vspace{-1em}
\end{wrapfigure}

Let   $P_T(u,v)$, the path between nodes $u$ and $v$ in  LLP tree $T$, have high latency inflation. Let $u'$ be a node on the path $P_T(u,v)$, starting at $u$ but before reaching the common ancestor $w$ of $u$ and $v$ in $T$.  A shortcut $(u’, v)$, creates a new tree path from node $u$ to node $v$,  that deviates from path $P_T(u,v)$ by following the shortcut at node $u'$  to directly get to node $v$ (Figure~\ref{fig:shortcut_example}).  And as the shortcut $(u’, v)$ uses the least latency underlay path between nodes $u’$ and node $v$, this new tree path from $u$ to $v$ can have much lower latency inflation than the original path $P_T(u,v)$. Note that we can always find such $u’$, since for $u'$ = $u$, the latency inflation from $u$ to $v$ drops to zero.

A shortcut $(u’, v)$ can be implemented as follows. 
The node $u'$ maintains forwarding entries to $v$ for the GIPs located at $v$. Packets forwarded on the tree to these GIPs, upon reaching the node $u’$, get directly forwarded to the node $v$ using the shortcut $(u’, v)$ . Although shortcuts help lower latency inflation, they increase the amount of forwarding state that needs to be maintained among the tree nodes (i.e., at the node $u’$  in the example above). Also, as forwarding state for a GIP 
may get replicated across many more nodes, there may be increased cost of updating node entries when a GIP changes location. 
We present below a solution that strikes a balance between the use of the least latency paths and the resulting cost of tracking mobile endpoints.

Let  $d$ and $D$ denote the smallest and largest latencies in $G$   respectively. We partition the interval $[d, D]$ into  $k$ contiguous latency ranges $\{ l_i = [d_{i-1}, d_i], 1 \le i \le k \}$, where $d_0 = d$ and $d_k = D$. We use latency inflation bound  $\epsilon_i$ for latency range $l_i$, with $ \epsilon_1 \le \epsilon_2 \le \ldots \le \epsilon_k$. Formally, this imposes the constraint  $LI(u,v) \le \epsilon_i$ (see Definition~\ref{inflation}), for any node pair $u, v$ whose direct path underlay latency $l(u,v)$  is in the range $l_i$.

{\bf Shortcut heuristic}: The shortcut heuristic algorithm starts with a tree $T$ and greedily augments it with shortcuts. It iterates over the latency ranges $l_1, l_2 , \ldots l_k$, in that order, starting from range $l_1$. For latency range $l_i$, it finds a pair of nodes $u, v$, with $l(u,v)$ in the range $l_i$, for which  latency inflation in the currently augmented tree is more than  $\epsilon_i$. It finds the “highest'' node (i.e., closest to the root of the tree) $u'$ of the type described earlier for which   shortcut $(u', v)$ brings the $u$ to $v$ latency inflation below $\epsilon_i$. 
The shortcut heuristic adds the shortcut $(u', v)$. It  repeats this step until there is no node pair $u, v$, with $l(u,v)$ in the range $l_i$,  for which the tree latency inflation is still greater than $\epsilon_i$. The heuristic then applies this procedure for latency range $l_{i+1}$ and so on until all  latency ranges have been covered.  

We now describe how CSR forwarding, location updates are handled in the shortcut augmented tree. We mainly focus on the modifications that specifically apply to the augmented tree $T$.

\subsubsection{CSR forwarding}
The only change is that when a   CSR, while it is being propagated up the tree, reaches a node $u$ that maintains a shortcut to leaf node $v$ for its GIP, then the node $u$ directly forwards the CSR to the node $v$, bypassing the rest of the tree.

\subsubsection{Location Updates} 
Each  leaf node $v$ of the   augmented tree $T$  maintains a set, $S_v$, which is the set of nodes of $T$ with a shortcut to $v$. That is  $u' \in S_v$, if shortcut $(u', v)$ is in $T$. 

When the node $v'$ receives the location update for the device  it uses the tree to forward it, as if it were a CSR, using  the destination GIP to the previous leaf node $v$. This message includes the set $S_{v'}$. When node $v$ receives this message, it sends an update message to all nodes in $S_v$ asking them to withdraw their current forwarding entry for this GIP that points to node $v$, while also sending a message to all nodes in $S_{v'}$ asking them to add a forwarding entry for this GIP that points to node $v'$. This ensures that the nodes' shortcut entries get updated to reflect the new location of the device.

\subsubsection{Analysis}
Note that by  design, the algorithm ensures worst case latency inflation bounds for the paths of the augmented tree. We now bound the expected number of  forwarding entries needed in $T$ for the shortcuts.

Recall  $d$ and $D$ denote the smallest and largest latencies in $G$   respectively. For the following result, we consider  latency ranges whose  boundaries are powers of two. Specifically    $\{ l_i = [d_{i-1}, d_i], 1 \le i \le k \}$, for $d_i = d 2^i, 0\le i \le k$. Note that since $d_k = D$, the number of different ranges $k$ is $O(\log \frac{D}{d})$. 

\begin{claim}
Let $n_i^v$ be the number of nodes $u$ such that $l(u,v)$ is in the range $l_i$. 
The expected number of shortcuts to $v$ and hence the number of forwarding entries for a GIP  located at a leaf node $v$  is at most $\sum_{i=1}^k \min \{ \frac{16n_i^v}{1+\epsilon_i} O(\log n), n_i^v \} $. 
\end{claim}

\begin{proof}
Let $l(u,v)$ be in the range $l_i$. 
Thus, $d 2^{i-1} \le l(u,v) < d 2^i$. Since a shortcut is added from $u$ to $v$,  $l_{T_{orig}}(u,v) > d 2^{i-1}  (1+\epsilon_i)$. Here $l_{T_{orig}}(u,v)$ is the $u$ to $v$  path latency in the LLP tree $T_{orig}$ constructed by HCS.
Let $w$ denotes the common ancestor of $u$ and $v$ in $T_{orig}$, at level $j$. Here, leaf nodes are at level $0$, their parents are at level $1$ and so on.
From Claim~\ref{latency-level}  it follows that $l_{T_{orig}}(u,v) \le d 2^{j+2}$. Thus, by combining bounds:
\begin{equation}
\vspace{-.75em}
d 2^{i-1}  (1+\epsilon_i) < l_{T_{orig}}(u,v)  <=  d 2^{j+2}
\end{equation}
This implies 
$$\frac{d 2^{i-1}}{ d 2^{j+2}} \le \frac{1}{1+\epsilon_i}  \  \mbox{or} \  \frac{d 2^i}{d2^{j-1}} \le \frac{16}{1+\epsilon_i}.$$
Also from ~\cite{metric2003fakcha} it follows that the probability that $w$ is at level $j$ of $T_{orig}$ is at most $\frac{l(u,v)}{d2^{j-1}} O(\log n)$. Since  $d 2^{i-1} \le l(u,v) < d 2^i$, this probability is at most 
\begin{equation}
 \frac{d 2^i}{d2^{j-1}} O(\log n) \le \min \{ \frac{16}{1+\epsilon_i} O(\log n), 1 \}
\end{equation}
Thus, the probability that the algorithm adds a shortcut from $u$ to $v$ is at most $\min \{ \frac{16}{1+\epsilon_i} O(\log n), 1 \}$. 
 \end{proof}

\section{Implementation}
\label{sec:implementation}
We implement our system using a combination of technologies. The data plane is realized with eBPF~\cite{calavera2019eBPF} by leveraging the tc~\cite{man2023tc_bpf} hook point in the Linux system, and the control plane is developed with bcc~\cite{bcc2023}.
We leverage Segment Routing IPv6 (SRv6) ~\cite{rfc8986} for packet forwarding within our proposed LLP tree structure.

CSR forwarding is a key component of our solution.
It leverages eBPF to intercept packets on virtual interfaces within GIP subnets. 
During the connection setup phase, an SRv6 header is inserted into the CSR packet, directing it through the LLP tree structure. This ensures the packet's destination IP address is ultimately replaced with the destination endpoint's PIP, guaranteeing successful delivery.
Our robust location update mechanism for mobility includes a link monitor component that detects a mobility event and sends update messages to remote hosts using an available interface, with the SRv6 header if needed.
This approach facilitates   seamless interface switching and maintains uninterrupted connections.

Moreover, in our scheme, we utilize the endpoint's own GIP as the destination for update messages.
This ensures these messages traverse the path containing outdated records, forcing them to be updated.
Furthermore, by having the  last LLP node on the path respond with an ACK, our scheme can signal a successful update within the LLP infrastructure. 
In the event that an ACK message is not received, the MN may retransmit the update.

\section{Evaluation}
\label{sec:eval}
We evaluate the performance of our LLP tree-based distributed solution from three key perspectives: latency inflation, memory consumption, location update disruption and its overhead.
We compare our solution with two other approaches: a centralized solution with a single mobility anchor responsible for managing all location updates and routing traffic to and from all mobile devices, and a LISP system with a centralized Map-Server storing all mapping records for all MNs. In our comparison, the centralized anchor or Map-Server is assumed to reside at the center of the LLP tree root node, as the LLP tree root is typically the central node in the topology based on latency.
LISP variants, such as LISP+ALT~\cite{lisp_alt}, LISP-TREE~\cite{lisp_tree}, and LISP-CONS~\cite{lisp_cons} include enhancements to improve the scalability and efficiency of the original LISP.
However, in high mobility scenarios their latency performance  may be worse, which is why we do not  compare our work to these variants.
For instance, LISP+ALT employs the prefix aggregability of BGP to disseminate  EID-prefixes, which can result in prolonged convergence times. LISP-TREE employs a DNS-like structure for  locator lookup, but maintaining a balance between mobility requirements and TTL for caching can be challenging. LISP-CONS works with CDN infrastructure, yet it still experiences long update latencies.

In our experiments, we simulate different solutions using real-world data to create realistic connection setups. We also utilize public Internet traceroute records to derive latency statistics between node pairs for various topologies. The datasets we used are described below:

\subsection{Datasets}
\label{sec:eval:datasets}

\BfPara{Internet topology zoo} 
We leverage the Internet Topology Zoo~\cite{knight2011itz}, a collection of real-world network topologies, to select $33$ topologies in the U.S. with adequate geolocation information for our experiments. Among them, the Arpanet19728 topology covers most of the U.S. geolocation and is used as a representative example of our evaluation results.

\BfPara{U.S. population} 
We use the population distribution to simulate realistic connection establishment events.
We utilize two datasets: one from ArcGIS hub~\cite{arcgis2023usa_counties}, which provides county population information, and another from Kaggle~\cite{kaggle2023counties_coordinates}, which lists the center point of each county.
To determine the endpoints of the connections in our simulation, we randomly select the center of a county based on its population size.

\BfPara{CAIDA Ark Dataset}
To obtain accurate latency measurements in our simulation, we use the  CAIDA Ark Dataset~\cite{ark_ipv4_traceroute}.
We select hops within the CAIDA Ark Dataset that are in proximity to our topologies' nodes.
Subsequently, we employ the traceroute latency of segments from the CAIDA Ark Dataset to derive pairwise latency for node pairs within our topologies. 
This approach allows us to closely mimic real-world latency conditions in our simulation.

\subsection{Latency Inflation of CSR}
In our LLP solution, all packets, except for Connection Setup Requests (CSRs), generally follow the direct least latency routes to reach their destinations, resulting in negligible latency inflation. The latency inflation of CSRs is therefore a crucial performance metric that we aim to assess. 

To evaluate the latency inflation of CSRs, we simulate connection setups from various locations across selected U.S.-based network topologies.
We randomly select two county centers as the connection endpoints and assume that the probability of establishing a connection is inversely proportional to the distance between the endpoints and directly proportional to the population density of the counties.
We also derive access latency from the county center to the topology nodes by dividing the geolocation distance with speed of light.

The results of this experiment are presented in Figure~\ref{fig:eval_scalability_csr_inflation_sorted}. 
The average latency inflation  for each topology is depicted in different overlapped bars: the green bar represents the tree path inflation, the red bar represents the centralized approach inflation, and the blue bar represents the LISP-based inflation.
The LISP-based inflation is evaluated under the assumption of no cache hits.
This assumption is  rooted in our emphasis on the CSR message, which is the first packet in the connection establishment process. In the context of LISP, cache misses are prevalent for the first packet, as this mechanism plays a crucial role in managing mobility. 
We can see that the LLP tree outperforms both the centralized approach and the LISP-based approach in all topologies, due to the shorter  packet paths along the LLP tree.
Across all topologies, the average LLP-based inflation for CSR is $0.0742$ or $7.42\%$.

\begin{figure}
\hspace*{\fill}
\begin{minipage}[t]{.63\linewidth}
  \centering
  \vspace{0em}
  \includegraphics[clip, width=0.95\columnwidth]{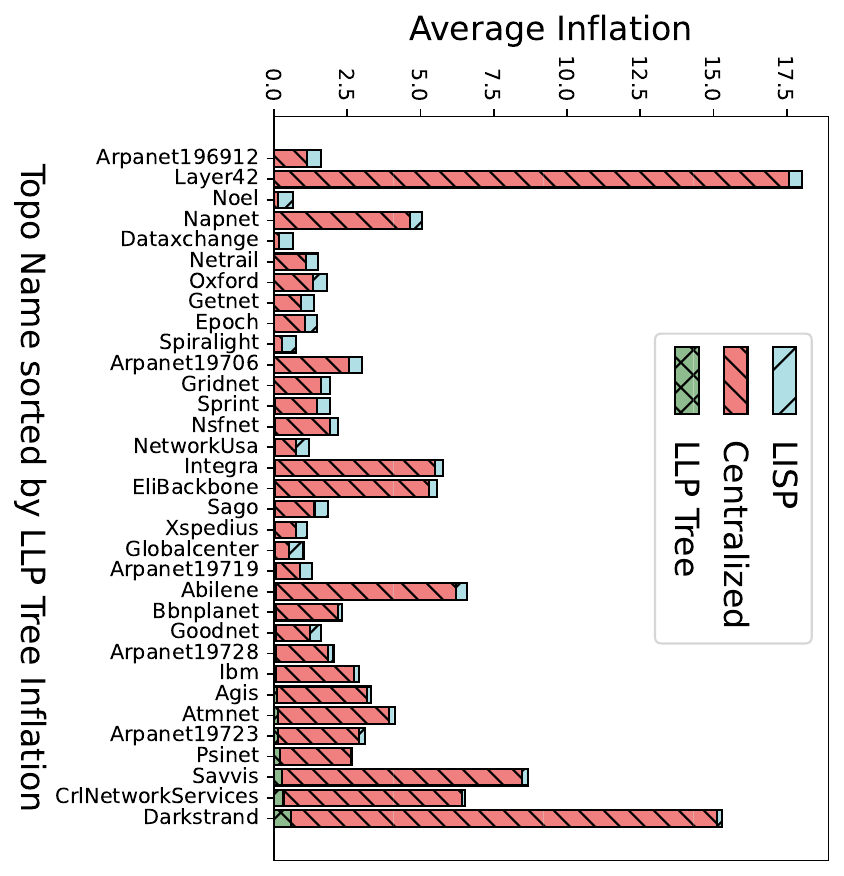}
  \captionof{figure}{Latency Inflation of CSR Processing}
  \label{fig:eval_scalability_csr_inflation_sorted}
\end{minipage}%
\hfill
\begin{minipage}[t]{.33\linewidth}
  \centering
  \vspace{0em}
  \includegraphics[clip, scale=0.30]{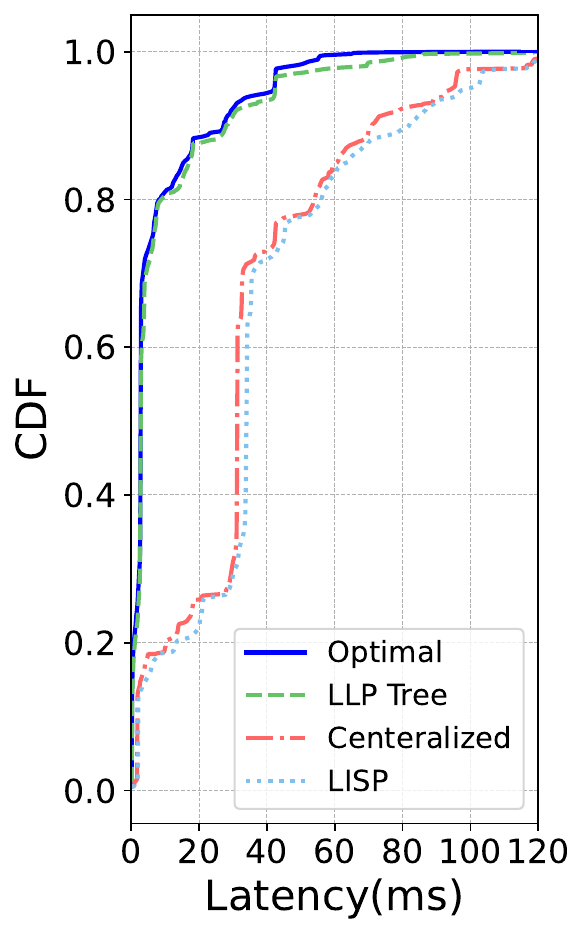}
  \vspace{-1em}
  \captionof{figure}{CSR Latency CDF}
  \label{fig:eval_scalability_csr_latency_cdf}
\end{minipage}
\hspace*{\fill}\vspace{-2em}
\end{figure}
Figure~\ref{fig:eval_scalability_csr_latency_cdf} displays the CSR latency CDF of various solutions for the Arpanet19728 topology, comparing LLP-based latencies with optimal latencies, centralized approach latencies and LISP-based approach latencies.
We can observe that the LLP-based latencies are closer to  optimal  latencies, while the latencies for the centralized approach are similar to the LISP-based approach due to cache misses causing LISP to behave similar to the centralized approach.

To better understand latency inflation in different situations, we categorize connections based on their direct path (the optimal) latencies and present the latency inflation for latency range, focusing on the low-latency categories where applications may be sensitive to latency changes.
Figure \ref{fig:eval_scalability_csr_latency_inflation_boxplot} compares the latency statistics of different solutions for each category, contrasting the LLP-based method, centralized approach, and LISP-based approach.
The LLP-based connection setup demonstrates stable low latency with low variance, while the centralized and LISP-based approach CSR latency is not only higher but also shows higher variance.

\begin{figure}[h]
\centering
    \includegraphics[clip,width=0.7\columnwidth]{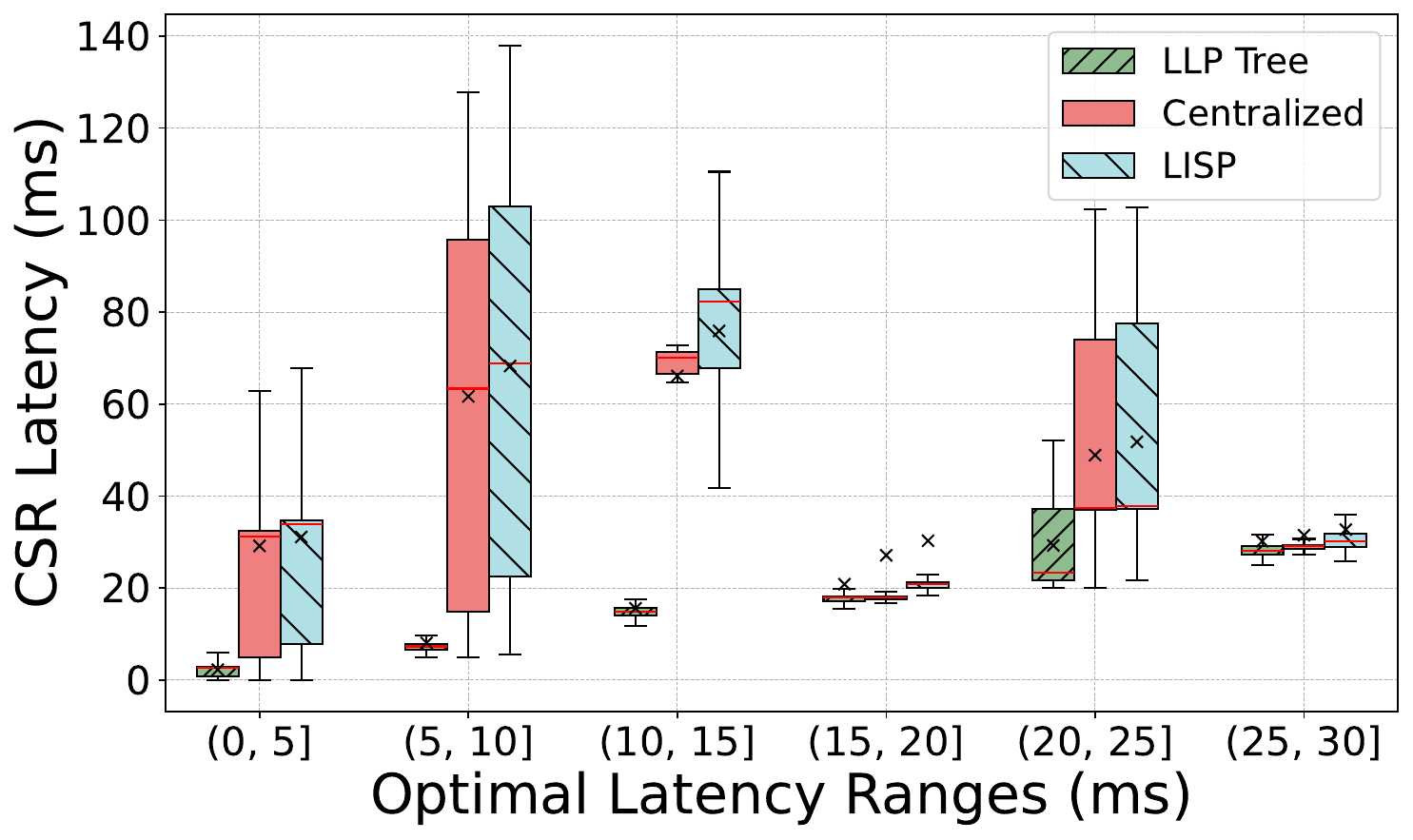}
    \caption{CSR Latency Inflation}
    \label{fig:eval_scalability_csr_latency_inflation_boxplot}\vspace{-0.5em}
    \vspace{0em}
\end{figure}

From the results, we can see that the LLP-based approach is effective in preventing latency increases compared to the centralized approach and LISP-based approach. 

\begin{figure*}[htbp]
\centering
    \begin{minipage}[m]{0.38\textwidth}
        \centering

        \includegraphics[clip, width=0.9\linewidth]{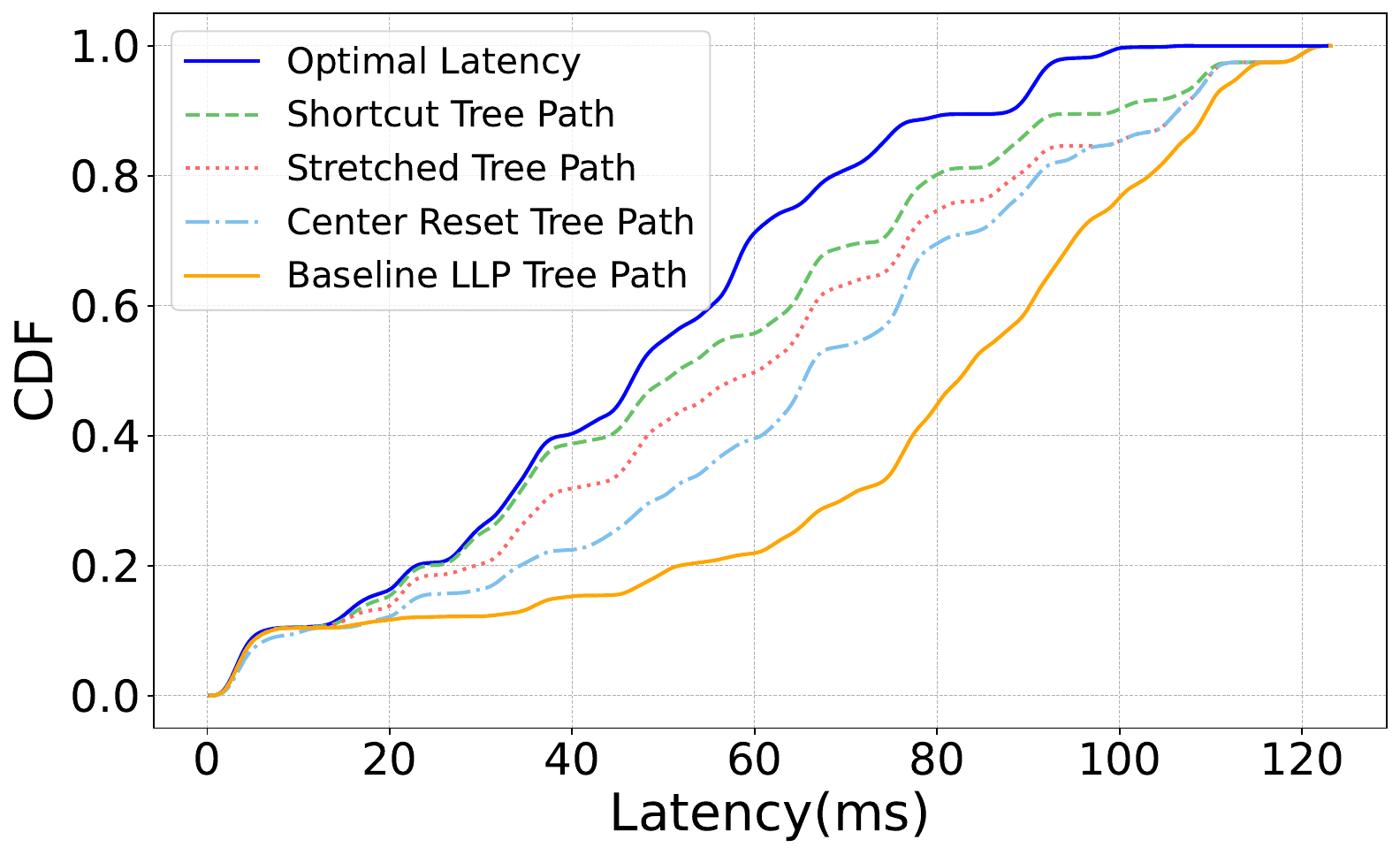}
        \caption{LLP Tree Heuristics: CSR Latency CDF}
        \vspace{3em}
        \label{fig:eval_heuristics_cdf}
    \end{minipage}
    \begin{minipage}[m]{0.27\textwidth}
        \centering
        \includegraphics[clip, width=\linewidth]{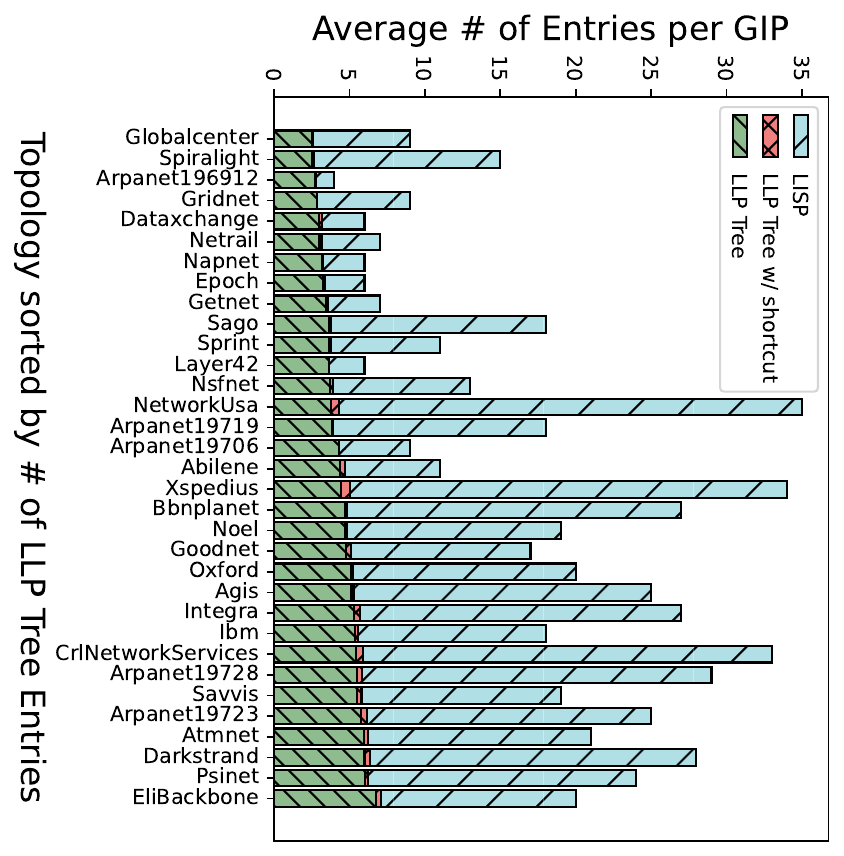}
        \caption{Memory Consumption over Topologies}
        \label{fig:eval_scalability_csr_replication_sorted}
    \end{minipage}
    \hspace{0em}
    \begin{minipage}[m]{0.27\textwidth}
        \includegraphics[clip, width=\linewidth]{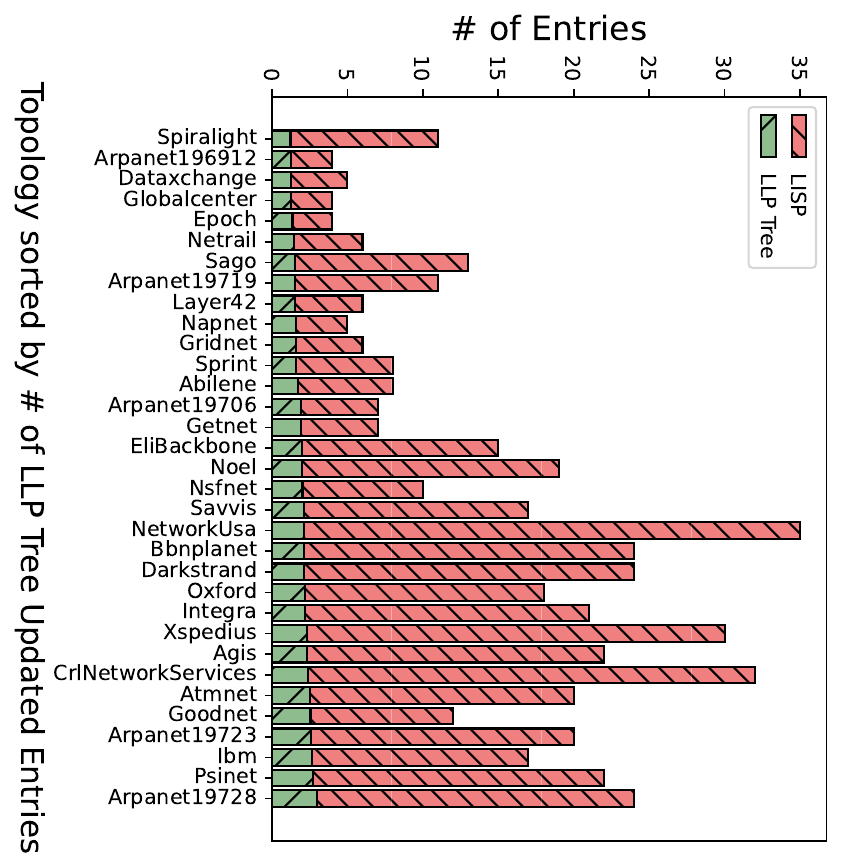}
      \caption{Number of LLP  Nodes Updated  Due to Mobility}
      \label{fig:eval_update_tree_vs_lisp}
    \end{minipage}

\vspace{-2em}
\end{figure*}

\BfPara{Improvement from heuristic algorithms}
\label{sec:eval:latency:improvement}
In this section, we build upon the baseline HCS algorithm to present the improvements in latency by the heuristic algorithms outlined in Section~\ref{improve-tree} and Section~\ref{sec:shortcut}.
Our evaluation is performed on the Arpanet19728 topology, and depict the CDF of latencies in Figure \ref{fig:eval_heuristics_cdf}.
Given that the described heuristics requires establishing specific conditions—such as maintaining fixed cluster centers when stretching the tree path—we apply the heuristics in the described order.
We observe the following:

\subsubsection{Cluster Center Reset}
We find an optimized cluster centers to balance the intra- and inter-cluster latencies in this step.
Figure~\ref{fig:eval_heuristics_cdf} shows CDF of this heuristic in light blue, which improves the LLP tree path latency.
On  average, we observe a  $15$\% improvement in CSR latency with this heuristic.

\subsubsection{Detours Elimination}
 Applied after the resetting cluster center heuristic, this heuristic further reduces the LLP tree path latency.
 Specifically, we observe an additional reduction in average path latency for CSRs by over $8.5$\%.

\subsubsection{Shortcut-based Bypass}
The shortcut heuristic limits the inflation of the tree, thereby improving latency performance. 
For our experiments, we  select $\epsilon=0.1$ for tree paths with latency lower than $10$ ms  and $\epsilon=1$ for all other tree paths.
We observe that the latency after applying the aforementioned heuristics is close to optimal.
The shortcut heuristic solely provides an average latency improvement of $8.1$\%.

To summarize, using these heuristics in sequence further improves the latency performance of the LLP tree.
After applying all optimizations to the Arpanet19728 topology, the CSR latency inflation is within $14.3$\% of the optimal latency.

\subsection{Memory Cost of Maintaining Entries}
We estimate the number of entries required for each mobile node (MN) in the LLP tree by the average tree depth, as only tree nodes from a leaf to its root maintain entries for an MN's GIP. 
We  compare these results to a fully distributed solution, where mapping entries for an MN are maintained  at every  node of the tree.
This worst-case scenario represents solutions such as LISP, where many network nodes (ITRs) may cache entries for an MN based on communication patterns.
To maintain a fair comparison with LISP, which only supports single access networks, we configure our solution to operate in single access mode as well. Consequently, our solution's entries store just one IP address per GIP, resulting in an overhead comparable to that of LISP, with no substantial increase in entry sizes.

For our evaluation, we build an LLP tree for each topology and compute its average tree depth and show the results in overlapped bars in Figure~\ref{fig:eval_scalability_csr_replication_sorted}. 
The results for LLP tree  are shown using green bars, and the worst-case results are shown using blue bars.
In our simulation, we assume that  traffic can  be initiated from anywhere within the network. 
For LISP, this implies that every network entry is potentially cached at each and every ingress node.
On average, our solution only needs to maintain $4.35$ entries per GIP, while in the LISP worst-case an average of $17.35$ entries are required.
This result shows our solution is $4$ times memory efficient than the LISP worst-case.

\BfPara{Shortcut memory overhead} 
In Section \ref{sec:shortcut}, we discussed a shortcut method to limit the latency inflation at the tree, which trades latency performance with more entries. 
In our simulation, we adopt the same $\epsilon$ values as in Section \ref{sec:eval:latency:improvement}.

To investigate the memory consumption introduced by the shortcut method, we simulated connections and calculated the extra entries needed over different topologies.
As shown in Fig.  \ref{fig:eval_scalability_csr_replication_sorted}, the entries needed by the shortcut mechanism are represented in red bars, and the extra entries introduced by the shortcut is very low.
This result shows adopting the shortcut heuristic is a reasonable trade-off, as it limits latency inflation without requiring a significant increase in memory consumption.

\subsection{Update Simulation}
\label{sec:eval:update_simulation}
\BfPara{Update disruption time}
In this section, we investigate the time costs to handle mobility events, since the disruption time caused by them is a critical factor in mobility management.
The disruption time includes the time to synchronize the location updates and the time for an end point to discover a mobility event.
We measured the time to discover a mobility event to be $38$ms on a Linux network namespace, where we turn off its access link in software to emulate the disruption.

In LISP, the location update process involves the mobile node (MN) registering with the new egress tunnel router (ETR) and unregistering with the old ETR.
The old ETR then sends notifications to the ingress tunnel router (ITR), causing the ITR to query the mapping server to synchronize.
On contrast, our solution allows the connecting endpoints to exchange mobility information directly through low-latency paths and handle the mobility event independently, without relying on LLP tree.

\begin{figure}[h]
\vspace{-1em}
\centering
  \includegraphics[trim={0 0 3cm 1.5cm},clip, width=0.7\columnwidth]{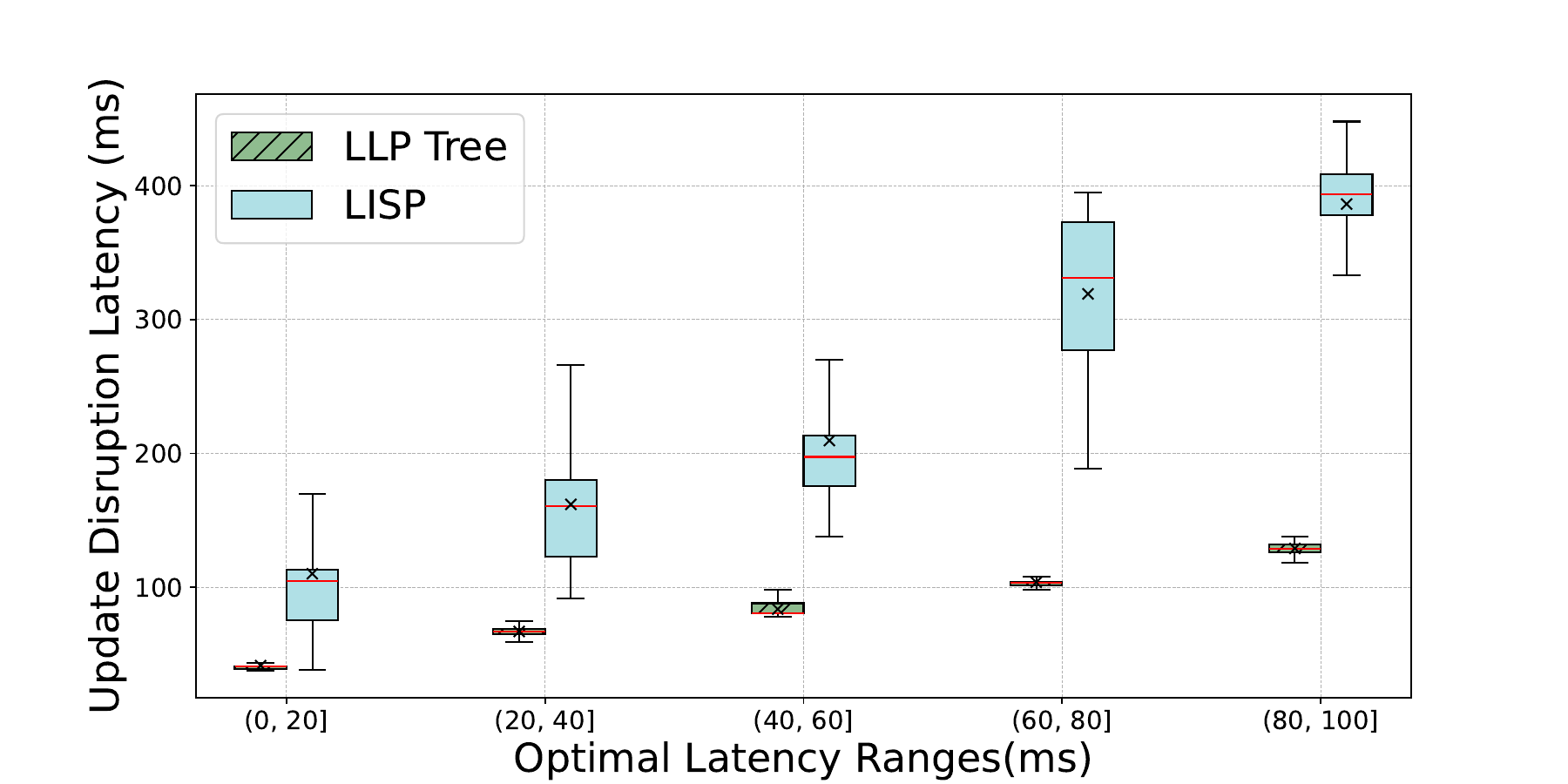}
  \caption{Update Disruption Latency: LLP Tree vs. LISP}
  \label{fig:eval_tree_vs_lisp}
  
  \vspace{-1em}
\end{figure}
We show the disruption time in Fig.~\ref{fig:eval_tree_vs_lisp}.
As a baseline, we group these  times by their direct path latency, which represents the optimal latency for completing the update in an ideal situation.
As can be seen in the figure, the average update disruption in LISP is $119.28$ms, while in our approach is $45.51$ms, showing that our solution is $2.6$ times faster than LISP in update. 
Our solution also shows lower variance in update disruption time than LISP in all categories.

\BfPara{Updated Nodes}
Along with the disruption time caused by mobility events, the number of nodes needs to be updated for this event is another critical performance metric.
We simulate location updates triggered by MN changing access points within a given topology.
Figure~\ref{fig:eval_update_tree_vs_lisp} displays the average number of LLP tree nodes that need updating for a mobility event in green bars.
In our approach, the average update involves fewer than $3$ nodes.
As a comparison, we investigate the worst case number of LISP nodes that need updating, and show the number in red bars.
One can find our approach updates much fewer nodes than LISP worst-case for handling a mobility event.
Moreover, in systems such as LISP, which depend on caching entries at the network's ingress points, numerous nodes may need updates even for local mobility.
These updates can impact a wide array of nodes distant from the MN.
In contrast, with the LLP tree solution, the effects are restricted to local areas.
This highlights the LLP tree scheme's enhanced efficiency in managing location updates.

\section{Related Work}

There are several IP mobility schemes that ensure session continuity during network transitions.  Prominent ones include Proxy Mobile IPv6 (PMIPv6)~\cite{rfc5213} and Mobile IPv6 (MIPv6)~\cite{mipv6}, which, despite their widespread use, can suffer from high latencies due to triangular routing via anchors. To address these latency issues, various route optimization proposals have been suggested, such as those in \cite{rfc6279, LRPMIP}. However, these are primarily applicable to situations where both the Mobile Node (MN) and Correspondent Node (CN) are connected to access gateways within the same provider domain or when data paths are optimized after session establishment.

 Distributed Mobility Management (DMM)~\cite{rfc8818, IEEE-DISTRIBUTED-MOBILITY}  schemes, such as Hierarchical Mobile IPv6 (HMIPv6)~\cite{rfc5380}, utilize local anchors at attachment points. However, these methods can incur high signaling costs for multiple anchors without significant latency reduction, or lack support for global mobility, causing connection breaks when the MN moves outside its anchor's local region.
Our solution, in contrast, ensures low latency for connection setup, is applicable regardless of MN and CN domain, and supports global mobility without excessive signaling overheads.

Approaches like LISP~\cite{rfc9300,rfc9301} and ILA~\cite{ila}   segregate Endpoint Identifiers  and Routing Locators,
with LISP using ingress tunnel routers to cache identifier-locator mappings, to help direct packets along least-latency paths. However, LISP's push mechanism~\cite{NERD} may not be scalable, while its pull model~\cite{costLISP} can lead to increased latency, disruption, and overheads due to cache misses and stale cache updates.

An improved BGP overlay based hybrid push/pull approach for LISP is proposed in~\cite{lisp_alt}.
Unlike our overlay solution, the  effectiveness and scalability  of this solution depends on   prefix aggregability of EIDs in local regions.
Furthermore, it may incur increased latency and more router entries due to reduced aggregability as devices become mobile.

A  DHT based LISP mapping lookup system is presented in~\cite{LISPDHT}.
In this case, DHTs are scalable but lookup latency can still be high.
As mentioned in Section~\ref{sec:eval},  LISP-TREE~\cite{lisp_tree}, and LISP-CONS~\cite{lisp_cons},  fail to keep up with highly dynamic update requirement. 
In contrast, our system enables end hosts to conduct access selection and switching, thereby facilitating rapid failovers and a broad spectrum of flow management policies. These policies can be optimized for various factors, including cost, energy usage, link quality, and network congestion.

Protocols such as QUIC \cite{quic, draft-ietf-quic-multipath-07} and MPTCP \cite{multipath_tcp} offer session continuity and multi-path data transmission capabilities at the transport layer. While they have proven valuable for mobility management~\cite{mptcpcell2012, cellbricks2021}, their customization for specific application needs often depends on individual applications, requiring varying degrees of complexity. Our system, in contrast, provides universally accessible multi-network utilization capabilities, for all applications using any IP protocol, not just QUIC and MPTCP, ensuring efficient data transmission across diverse network environments.

\section{Conclusion}
\label{conclusion}
We have presented a novel end-to-end solution for providing internet-scale low-latency mobility management for multi-connected devices. By leveraging the locator/ID separation principle and introducing a tree embedding-based overlay for dynamic locator lookup, our solution can provide near-optimal data forwarding and allow efficient utilization of multiple access networks available to a mobile device even under unexpected connectivity disruptions. Extensive analysis and simulations based on real network data show that our approach has significantly advantages over existing solutions. We are currently finalizing the implementation of the comprehensive end-to-end system to further validate our design and address potential deployment issues.

\section*{Acknowledgement}
We would like to thank the reviewers of INFOCOM'25 for their invaluable feedback. 
This work is partially supported by an NSF grant CNS-2008468 and an ONR grant N00014-23-1-2137. A major part of the work was done during the first author's internship at Nokia Bell Labs.

\clearpage

\printbibliography

\end{document}